\documentclass[submission,copyright,creativecommons]{eptcs}
\providecommand{\event}{AiML 2026} 

\usepackage{aiml26,iftex}

\ifpdf
  \usepackage{underscore}         
  \usepackage[T1]{fontenc}        
\else
  \usepackage{breakurl}           
\fi

\usepackage{amsmath, amssymb, wrapfig, caption, mdframed}
\usepackage{multicol}
\usepackage[numbers]{natbib}

\usepackage[shortlabels]{enumitem}

    \input xy
	\xyoption{all}

\newcommand{\lang}[1]{\mathbf{#1}}
\newcommand{\logic}[1]{\mathsf{#1}}

\newcommand{\positions}[1]{\mathsf{#1}} 
\newcommand{\At}{\mathsf{At}}

\newcommand{\R}{\mathbin{\mathsf{R}}}

\theoremstyle{definition}

\title{Hyperformalism for Relevant Modal Logics}
\author{Thomas Macaulay Ferguson
\institute{Rensselaer Polytechnic Institute \\
Department of Cognitive Science\\
Troy, New York, USA}
\email{tferguson@gradcenter.cuny.edu}
\and
Shay Allen Logan
\institute{Kansas State University\\
Department of Philosophy\\
Manhattan, Kansas, USA}
\email{salogan@ksu.edu}
}

\newcommand{\titlerunning}{Hyperformalism for Relevant Modal Logics}
\newcommand{\authorrunning}{T. M. Ferguson \& S. A. Logan}

\hypersetup{
  bookmarksnumbered,
  pdftitle    = {\titlerunning},
  pdfauthor   = {\authorrunning},
  pdfsubject  = {Hyperformalism for Relevant Modal Logics},               
}

\begin{document}
\maketitle

\begin{abstract}
The property of \emph{hyperformalism} has proven to be a powerful tool in the analysis of relevant logics, revealing that increasingly weak relevant logics are closed under increasingly strong classes of non-uniform substitutions. In such substitutions, two instances of the same atom may be treated independently in virtue of syntactic features of their appearances in a complex. In this work, we extend the scope of hyperformalism to relevant modal logics by considering \textsf{MPos}-hyperformalism, that is, a property in which relevant modal logics are closed under substitutions in which nesting within the scope of modal operators is taken into account. We prove that the weak relevant modal logic $\logic{B}^\Box$ is \textsf{MPos}-hyperformal and investigate the classes of non-uniform substitutions under which several extensions are closed. We then consider corresponding refinements of the variable sharing property that hold of such logics. We conclude by introducing a modal logic $\logic{K}^{\positions{MPos}}$ that constitutes the largest \textsf{MPos}-hyperformal sublogic of the classical modal logic $\logic{K}$ and provide soundness and completeness results.
\end{abstract}

\section{Introduction}

In typical formal logics, atomic formulas are treated as being entirely independent. This is (at least partly, and maybe wholly) captured by their being closed under uniform substitutions, since closure under uniform substitutions means that one can safely replace occurrences of distinct atoms with distinct formulas without needing to worry about in any way coordinating these substitutions. 

\textit{Hyperformal} logics take the notion of independence captured by uniform substitutions and run with it. More concretely (but still abstractly) a logic is hyperformal when occurrences of the same atom in sufficiently different syntactic contexts can be treated as independent. As with uniform substitutions and independence for atoms, the independence enforced by hyperformalism is made concrete in a logic $\logic{L}$ by the requirement that $\logic{L}$-theoremhood be preserved under non-uniform substitutions treating such instances independently. \emph{E.g.}, \cite{logan2022depth} shows that the relevant logic $\logic{DR}$ is closed under substitutions that treat atoms occurring in the scope of different numbers of conditionals---which is to say, that occur at different \emph{depths}---as independent. For example: in the formula $p\rightarrow(q\rightarrow p)$, the two instances of the atom $p$ appear at different depths and are thus are treated as independent, \emph{i.e.}, as though the they were distinct atoms. Consequently, the formula $p\rightarrow(q\rightarrow r)$ follows from a depth substitution. That the latter formula lacks any variable common to its antecedent and consequent ensures that it is not a theorem of $\logic{DR}$, whence $p\rightarrow(q\rightarrow p)$ is not a theorem. Increasingly fine-grained species of hyperformalism have since been introduced (see \emph{e.g.} \cite{ferguson2023topic,standefer2024topics})---along with corresponding fine-grained variable sharing properties---proving hyperformalism to be a powerful tool in the analysis of relevant logics.\footnote{For more on relevant logics see e.g. \cite{logan2024relevance}, \cite{Mares2004}, \cite{standefer2024routes}, or \cite{read1988relevant}.}

Relevant \emph{modal} logics have nearly as long a history as relevant logics themselves.\footnote{For important early papers, see e.g. \cite{routleymeyer1972} or \cite{fuhrmann1990}; more recent development can be found in e.g. \cite{seki2011}, \cite{standefervarietiesS5}, or \cite{ferenz2025}.} Given how closely the analysis of modality is tied to the broader relevant logic project, it is natural to consider species of hyperformalism in which occurrences of the modal $\Box$ have the same weight as $\rightarrow$ enjoys in the case of $\mathsf{DR}$. The results in this paper will thus concern extending the notion of hyperformalism to modal logics whose languages include $\Box$. We will begin in Section \ref{section:BBox} by considering modal extensions of one of the weakest standard relevant logics, $\logic{B}$, and investigating the types of strong hyperformal properties that hold of members of this family of modal logics, including what we will call $\mathsf{MPos}$-hyperformalism. We will take an alternative tack in Section \ref{section:KMPos} by describing and characterizing the largest $\mathsf{MPos}$-hyperformal fragment of the modal logic $\logic{K}$ through a tableaux system. We will conclude by considering the relationship between the two approaches and the prospects for variants of hyperformalism in the modal logic setting.

\section{$\logic{B}^{\Box}$ and Its Extensions}\label{section:BBox}

We begin by considering hyperformalism in the context of familiar relevant modal logics, starting from the base logic $\logic{B}^{\Box}$. To introduce the logic $\logic{B}^{\Box}$, let $\lang{L}$ be a propositional language formed in the usual way from the set of atomic formulas $\At$ using the connectives $\land$, $\lor$, $\neg$, and $\to$. Then $\logic{B}$ is axiomatized by the following list. 

\begin{multicols}{2}
    \begin{enumerate}[{BA}1.]
        \item $A\to A$
        \item $(A\land B)\to A$; $(A\land B)\to B$
        \item $A\to(A\lor B)$; $B\to(A\lor B)$
        \item $((A\to C)\land (B\to C))\to((A\lor B)\to C)$
        \item $((A\to B)\land (A\to C))\to(A\to (B\land C))$
        \item $(A\land(B\lor C))\to((A\land B)\lor(A\land C))$
        \item $\neg\neg A\to A$
    \end{enumerate}

    
    \begin{enumerate}[{BR}1.]
        \item $A, B\Rightarrow A\land B$
        \item $A\to B, A\Rightarrow B$
        \item $A\to\neg B\Rightarrow B\to\neg A$
        \item $A\to B, C\to D\Rightarrow(B\to C)\to(A\to D)$
    \end{enumerate}

\end{multicols}

\noindent Modal extensions of $\logic{B}$ follow the addition of one or more of the following axioms or rules; the base system with $\logic{B}^{\Box}$ defined by adding C and RM to $\logic{B}$:\footnote{We preserve the standard naming conventions as found in \emph{e.g.} \cite{chellas1980}.}

\begin{multicols}{2}

    \begin{enumerate}
        \item[C] $(\Box A\land\Box B)\to\Box(A\land B)$
        \item[T] $\Box A\to A$
        \item[4] $\Box A \to \Box\Box A$
        \item[K] $\Box(A\to B)\to(\Box A\to\Box B)$
    \end{enumerate}

      \begin{enumerate}
        \item[RM] $A\to B~\Rightarrow~\Box A\to\Box B$
        \item[RN] $A\Rightarrow\Box A$ 
    \end{enumerate}

\end{multicols}


\subsection{Non-Uniform Substitutions in Relevant Logics}

$\logic{B}$, along with a few other relevant logics, exhibits a few forms of hyperformalism worth describing. But we wil go about this in a mildly idiosyncratic way, since it will serve our aims below to do so. 

To begin, let $\positions{CPos}=\{\mathtt{c}\}^*$ be the Kleene closure of the set containing just $\mathtt{c}$.\footnote{Usually, instead of $\positions{CPos}$, one uses natural numbers and refers to these numbers as `depths'; the numerical value of depth, clearly, can be recovered by counting the instances of $\mathtt{c}$ in the position. See e.g \cite{logan2022depth,leach2024logic,WorleyForthcoming}. Sets that play the role that $\positions{CPos}$ plays here will be called sets of \textit{positions} below.} Say that a function $\sigma:\positions{CPos}\times\At\to\lang{L}$ is a $\positions{CPos}$-substitution on $\lang{L}$. $\positions{CPos}$-substitutions naturally extended to act on all of $\positions{CPos}\times\lang{L}$ via the recursive clauses in Figure~\ref{initialsubs}. A set of formulas $X$ is $\positions{CPos}$-hyperformal when for all $\positions{CPos}$-substitutions $\sigma$, if $A\in X$, then $\sigma_{\varepsilon}(A)\in X$. As discussed in the introduction, note that in $\positions{CPos}$-hyperformal logics, atoms occuring `at' different $\positions{CPos}$-labels in a formula are treated as independent to the extent that they can be replaced by different formulas with no worry about coordinating between the different parts of this replacement.

It turns out, as first shown in \cite{logan2022depth}, that $\logic{B}$ is $\positions{CPos}$-hyperformal. This fact is the first of several hyperformalism results one can now find in the literature (see \cite{standefer2024topics, logan2024hyperformalism, ferguson2023topic} for other such results). To describe another (much more powerful) such result, we first define $X^{*\mathtt{c}}$ to be $X^*\cup\{\mathtt{\overline{x}c\mid\overline{x}}\in X^*\}$ and let $\positions{LRCNPos}=\mathtt{\{l, r, n\}}^{*\mathtt{c}}$. 

What we just did with $\positions{CPos}$, we can now do with $\positions{LRCNPos}$. Explicitly, say $\positions{LRCNPos}$-substitutions are functions $\sigma:\positions{LRCNPos}\times\At\to\lang{L}$ for which $\sigma_{\mathtt{\overline{x}nn\overline{y}}}(p)=\sigma_{\mathtt{\overline{xy}}}(p)$ for all $\{\overline{\mathtt{x}},\overline{\mathtt{y}}\}\subseteq\positions{LRCNPos}$ and $p\in\At$.\footnote{This is a feature described in \cite{standefer2024topics} as ``faithfulness.''} We extend such a function to all of $\positions{LRCNPos}\times\lang{L}$ using the recursive clauses in Figure~\ref{initialsubs}. Say that a set of formulas $X$ is $\positions{LRCNPos}$-hyperformal when for all $\positions{LRCNPos}$-substitutions $\sigma$, if $A\in X$, then $\sigma_{\varepsilon}(A)\in X$. As shown in \cite{standefer2024topics}, it turns out that $\logic{B}$ is \emph{also} $\positions{LRCNPos}$-hyperformal.

\begin{figure}[t]
    \begin{multicols}{2}
        \mbox{}\vfill
        \noindent $\sigma_{\overline{\mathtt{x}}}(A\land B)=\sigma_{\overline{\mathtt{x}}}(A)\land \sigma_{\overline{\mathtt{x}}}(B)$\smallskip
        
        \noindent $\sigma_{\overline{\mathtt{x}}}(A\lor B)=\sigma_{\overline{\mathtt{x}}}(A)\lor \sigma_{\overline{\mathtt{x}}}(B)$\smallskip

        \noindent $\sigma_{\overline{\mathtt{x}}}(\neg A)=\neg \sigma_{\overline{\mathtt{x}}}(A)$\smallskip
		
        \noindent $\sigma_{\overline{\mathtt{x}}}(A\to B)=\sigma_{\mathtt{c\overline{x}}}(A)\to \sigma_{\mathtt{c\overline{x}}}(B)$
        \vfill\columnbreak

        \noindent $\sigma_{\overline{\mathtt{x}}}(A\land B)=\sigma_{\overline{\mathtt{x}}}(A)\land\sigma_{\overline{\mathtt{x}}}(B)$\smallskip
    
    	\noindent $\sigma_{\overline{\mathtt{x}}}(A\lor B)=\sigma_{\overline{\mathtt{x}}}(A)\lor\sigma_{\overline{\mathtt{x}}}(B)$\smallskip
    
    	\noindent $\sigma_{\overline{\mathtt{x}}}(\neg A)=\neg \sigma_{\mathtt{n}\overline{\mathtt{x}}}(A)$\smallskip
    
    	\noindent $\sigma_{\varepsilon}(A\to B)=\sigma_{\mathtt{c}}(A)\to\sigma_{\mathtt{c}}(B)$\smallskip
    
    	\noindent $\sigma_{\overline{\mathtt{x}}}(A\to B)=\sigma_{\mathtt{l}\overline{\mathtt{x}}}(A)\to\sigma_{\mathtt{r}\overline{\mathtt{x}}}(B)$ for $\overline{\mathtt{x}}\neq\varepsilon$
    \end{multicols}
\caption{$\positions{CPos}$-substitutions on the left; $\positions{LRCNPos}$-substitutions on the right.}\label{initialsubs}
\end{figure}

\begin{wrapfigure}[11]{R}{3.25in}
    \begin{mdframed}
        \noindent $\alpha_{\overline{\mathtt{x}}}(A\land B)=\min(\alpha_{\overline{\mathtt{x}}}(A), \alpha_{\overline{\mathtt{x}}}(B))$;\smallskip
    
        \noindent $\alpha_{\overline{\mathtt{x}}}(A\lor B)=\max(\alpha_{\overline{\mathtt{x}}}(A), \alpha_{\overline{\mathtt{x}}}(B))$;\smallskip
            
        \noindent $\alpha_{\overline{\mathtt{x}}}(\neg A)=1-\alpha_{\mathtt{n}\overline{\mathtt{x}}}(A)$;\smallskip
        
        \noindent $\alpha_{\varepsilon}(A\to B)=\max(1-\alpha_{\mathtt{c}}(A), \alpha_{\mathtt{c}}(B))$; and\smallskip
        
        \noindent $\alpha_{\overline{\mathtt{x}}}(A\to B)=\max(1-\alpha_{\mathtt{l\overline{x}}}(A), \alpha_{\mathtt{r\overline{x}}}(B))$ for $\overline{x}\neq\varepsilon$.
    \end{mdframed}
    \caption{$\positions{LRCNPos}$-assignments}\label{LRCNclauses}
\end{wrapfigure}

$\positions{LRCNPos}$-hyperformalism is interesting in other ways too. To see this, let an $\positions{LRCNPos}$-assignment be a function $\alpha:\positions{LRCNPos}\times\At\to\{0, 1\}$. $\positions{LRCNPos}$-assignments extend to all of $\positions{LRCNPos}\times\lang{L}$ by the recursive clauses given in Figure~\ref{LRCNclauses}. Say that $A$ is $\positions{LRCNPos}$-valid (and we write $\vDash_{\positions{LRCNPos}} A$) just if $\alpha_{\varepsilon}(A)=1$ for all $\positions{LRCNPos}$-assignments $\alpha$. $\positions{LRCNPos}$-validity can be tracked by a fairly obvious tableau system similar to the one we present below. In \cite{standefer2024topics}, the authors show not only that this tableau system is sound and complete for $\positions{LRCNPos}$-validity, but that the logic it captures is the largest $\positions{LRCNPos}$-hyperformal subset of classical logic and---more surprising yet---that it has the variable sharing property. 

So not only does hyperformalism point us in the direction of interesting features of known relevant logics, it also points us in the direction of previously unknown relevant logics with simple semantic theories as well. Our goal in this paper is to extend the techniques used in these results to the case of modal logics.

Following the guide in the introduction, we begin by examining the forms of hyperformalism exhibited by known relevant modal logics.\footnote{For more on relevant modal logics see \cite{fine1974models} and \cite{fuhrmann1990}.} We will begin our investigation by focusing first on $\logic{B}^\Box$, and will turn to stronger relevant modal logics in the subsequent section.

To define hyperformalism for modal relevant logics, we first need to do a bit of setup. Our initial set of positions is $\positions{MPos}=\mathtt{\{l, r, n, b\}}^{*\mathtt{c}}$. An $\positions{MPos}$-substitution is a function $\positions{MPos}\times\At\to\lang{L}$ for which $\sigma_{\mathtt{\overline{x}nn\overline{y}}}(p)=\sigma_{\mathtt{\overline{xy}}}(p)$ for all $\{\overline{\mathtt{x}},\overline{\mathtt{y}}\}\subseteq\positions{MPos}$ and $p\in\At$. We extend $\positions{MPos}$-substitutions to act on all of $\lang{L}$ by the recursive clauses in Figure~\ref{mpossubs}. 

\begin{figure}[t]
    \begin{multicols}{2}
    \begin{itemize}[leftmargin=*]
        \item $\sigma_{\overline{\mathtt{x}}}(A\land B)=\sigma_{\overline{\mathtt{x}}}(A)\land \sigma_{\overline{\mathtt{x}}}(B)$;        
        \item $\sigma_{\overline{\mathtt{x}}}(A\lor B)=\sigma_{\overline{\mathtt{x}}}(A)\lor \sigma_{\overline{\mathtt{x}}}(B)$;
	    \item $\sigma_{\overline{\mathtt{x}}}(\neg A)=\neg \sigma_{\mathtt{n}\overline{\mathtt{x}}}(A)$;
        \item $\sigma_{\varepsilon}(A\to B)=\sigma_{\mathtt{c}}(A)\to\sigma_{\mathtt{c}}(B)$;
	    \item \mbox{$\sigma_{\overline{\mathtt{x}}}(A\to B)=\sigma_{\mathtt{l\overline{x}}}(A)\to \sigma_{\mathtt{r\overline{x}}}(B)$} for $\overline{\mathtt{x}}\neq\varepsilon$; and
	    \item $\sigma_{\overline{\mathtt{x}}}(\Box A)=\Box \sigma_{\mathtt{b}\overline{\mathtt{x}}}(A)$.
    \end{itemize}
    \end{multicols}
\caption{$\positions{MPos}$-substitutions.}\label{mpossubs}
\end{figure}

\noindent We then say that a logic is $\positions{MPos}$\emph{-hyperformal} if its set of theorems is closed under $\positions{MPos}$-substitutions.

\subsection{The Base System}

As promised, we begin by examining $\logic{B}^{\Box}$. Given the preceding definitions, we can quickly demonstrate the main result concerning $\logic{B}^{\Box}:$




\begin{theorem}\label{theorem:BBox-Hyperformal}
    $\logic{B}^\Box$ is $\positions{MPos}$-hyperformal.
\end{theorem}

\begin{proof}
    By induction on derivations. The details for the axioms and rules of $\logic{B}$ follow from arguments virtually identical to those provided for Theorems 2 and 3 of \cite{standefer2024topics}, so we consider the axiom C and rule RM. 

   For C, we must show that every substitution applied to an instance of this axiom is itself provable. For an arbitrary substitution $\sigma$, some manipulation reveals that $\sigma_{\varepsilon}((\Box A\wedge\Box B)\rightarrow\Box(A\wedge B))$ is $(\Box \sigma_{\mathtt{bc}}(A)\wedge\Box\sigma_{\mathtt{bc}}(B))\rightarrow \Box(\sigma_{\mathtt{bc}}(A)\wedge\sigma_{\mathtt{bc}}(B))$. But this is another instance of C, so provable.
    
   For RM, we must show that for an arbitrary substitution $\sigma$ that if $\Box A\rightarrow\Box B$ is a theorem due to an application of RM, then $\sigma_{\varepsilon}(\Box A\rightarrow\Box B)$ is a theorem as well. So suppose that $\Box A\rightarrow \Box B$ is a theorem of $\logic{B}^{\Box}$ yielded by the application of RM to an earlier theorem $A\rightarrow B$. Fix an arbitrary $\sigma$ and define an further substitution $\sigma^{\star}$ by the following scheme:

    \begin{center}
        $\sigma^{\star}_\mathtt{\overline{x}}(p)=
        \begin{cases}
            \sigma_{\mathtt{\overline{y}bc}}(p) & \mbox{ if }\mathtt{\overline{x}}=\mathtt{\overline{y}c}\\
            \sigma_{\mathtt{\overline{x}}}(p) & \mbox{ otherwise}
        \end{cases}$
    \end{center}
    
    \noindent Clearly, $\sigma^{\star}$ is faithful and is thus itself a substitution. Moreover, a trivial induction shows that $\sigma^*_{\mathtt{\overline{x}}}(A)=\sigma_{\mathtt{\overline{y}bc}}(A)$ when $\mathtt{\overline{x}=\overline{y}c}$ for arbitrary (not necessarily atomic) formulas $A$. Consequently, by induction hypothesis, $\sigma^{\star}_{\varepsilon}(A\rightarrow B)$ is a theorem, whence $\sigma^{\star}_{\mathtt{c}}(A)\rightarrow\sigma^{\star}_{\mathtt{c}}(B)$ is a theorem, whence $\sigma_{\mathtt{bc}}(A)\rightarrow\sigma_{\mathtt{bc}}(B)$ is a theorem. Take a derivation of this formula and apply RM to yield $\Box\sigma_{\mathtt{bc}}(A)\rightarrow\Box\sigma_{\mathtt{bc}}(B)$. With some manipulation, this establishes theoremhood of $\sigma_{\varepsilon}(\Box A\rightarrow\Box B)$; as $\sigma$ was arbitrary, we infer that RM preserves hyperformalism.
\end{proof}

\noindent Now, we can turn to the matter of the types of variable sharing properties that are enjoyed by $\logic{B}^{\Box}$. As in \cite{logan2022depth,ferguson2023topic,standefer2024topics}, non-uniform substitutions can be applied to show that $\logic{B}^{\Box}$ enjoys very strong versions of the traditional hallmark of relevant logics. First, consider some definitions and lemmas:

\begin{definition}
    For two positions $\mathtt{\overline{x}}$ and $\mathtt{\overline{y}}$, say that $\mathtt{\overline{x}}\sim_{\mathtt{nn}}\mathtt{\overline{y}}$ if one can produce one from the other by the addition or deletion of instances of the string $\mathtt{nn}$. Say that a substitution $\sigma$ is \emph{atomic} in case its range is $\mathsf{At}$ and \emph{injective up to faithfulness} in case $\sigma_{\mathtt{\overline{x}}}(A)=\sigma_{\mathtt{\overline{y}}}(B)$ holds if and only if $\mathtt{\overline{x}}\sim_{\mathtt{nn}}\mathtt{\overline{y}}$ and $A=B$.
\end{definition}

\noindent The ``faithfulness'' of the equivalence relation $\sim_{\mathtt{nn}}$ rests in its assurance that axioms reflecting double negation introduction and elimination are preserved. In considering the results for modal logics based on weak relevant logics in which axioms like $A\rightarrow\neg\neg A$ fail---like $\logic{BM}$---one might be interested in relaxing the assumption of faithfulness. Note also that in proving Theorem \ref{theorem:BBox-Hyperformal}, we had assumed a condition that $\sigma_{\overline{\mathtt{x}}\mathtt{nn}\overline{\mathtt{y}}}(p)=\sigma_{\overline{\mathtt{xy}}}(p)$ for all atoms $p$, considerations of elegance lead us to relax the condition and assume only injectivity up to faithfulness instead.

For a formula $A$ and instance of a subformula $B$ appearing in $A$, we define the position of $B$ in $A$---$\mathsf{pos}(A[B])$---to be the position in (element of) $\mathsf{MPos}$ corresponding to the syntactic path through which one arrives at $B$ within $A$. For example, $\mathsf{pos}( \Box(p\rightarrow\Box q)[q])=\mathtt{brb}$. (For more precise recursive details, the reader is invited to adapt the steps described in \cite{ferguson2023topic} or \cite{standefer2024topics}.)

We will first state the strongest $\positions{MPos}$-based variable sharing property available in the case of $\mathsf{B}^{\Box}$.

\begin{definition}
A logic in the language $\mathbf{L}$ enjoys the $\mathsf{MPos}$-variable sharing property in case whenever $A\rightarrow B$ is provable, there are $A[p]$ in $A$ and $B[p]$ in $B$ such that $\mathsf{pos}(A[p]\rightarrow B)\sim_{\mathtt{nn}}\mathsf{pos}(A\rightarrow B[p])$.
\end{definition}

\noindent To prove this, we first observe the following:

\begin{lemma}\label{lemma:vanilla-vsp}
    $\logic{B}^{\Box}$ has the variable sharing property.
\end{lemma}

\begin{proof}
    This follows from Theorem 4.4 of \cite{standefervarietiesS5}, which shows that the relevant modal logic $\logic{RS5}$ has the variable sharing property. Insofar as $\logic{B}^{\Box}$ is a subsystem of $\logic{RS5}$, any theorem of the former is a theorem of the latter and thus must have the variable sharing property.
\end{proof}

\noindent Lemma \ref{lemma:vanilla-vsp} provides the necessary groundwork to begin the following theorem:

\begin{theorem}\label{theorem:BBox-MPosVSP}
    $\logic{B}^{\Box}$ has the $\mathsf{MPos}$-variable sharing property.
\end{theorem}

\begin{proof}
    Suppose that $A\rightarrow B$ is a theorem of $\logic{B}^{\Box}$. Choose a function $\mathsf{g}$ from pairs $\langle\mathtt{\overline{x}},p\rangle\in \mathsf{MPos}\times\mathsf{At}$ to $\At$ such that $\mathsf{g}(\langle\mathtt{\overline{x}},p\rangle)=\mathsf{g}(\langle\mathtt{\overline{y}},q\rangle)$ iff $\mathtt{\overline{x}\sim_{nn}\overline{y}}$ and $p=q$. Define the $\mathsf{MPos}$-substitution $\sigma$ by $\sigma_{\mathtt{\overline{x}}}(p)=p_{\mathsf{g}(\langle\mathtt{\overline{x}},p\rangle)}$. By Theorem \ref{theorem:BBox-Hyperformal}, $\sigma_{\mathtt{c}}(A)\rightarrow\sigma_{\mathtt{c}}(B)$ is a theorem of $\logic{B}^{\Box}$. By Lemma \ref{lemma:vanilla-vsp}, $\sigma_{\mathtt{c}}(A)$ and $\sigma_{\mathtt{c}}(B)$ share some variable $p_{k}$, but consideration of the definition of $\mathsf{g}$ entails that an atom $p_{k}$ appears in both $A$ and $B$ only in case $\mathsf{pos}(A[p_{k}])\sim_{\mathtt{nn}}\mathsf{pos}(B[p_{k}])$, witnessing the $\mathsf{MPos}$-variable sharing property.  
\end{proof}

\noindent This also allows the description of some more simply-stated---but nevertheless strong---versions of the variable sharing property.

\begin{corollary}\label{corollary:boxes}
    If $A\rightarrow B$ is a theorem of $\logic{B}^{\Box}$, then $A$ and $B$ share a propositional variable that is nested within the same number of instances of $\Box$ in both $A$ and $B$.
\end{corollary}

\noindent If one, like  \cite{standefer2024topics}, interprets relevance in terms of subject-matter,  Lemma \ref{lemma:vanilla-vsp} constitutes a reflection of the discussion of \cite{ferguson2023c}, that is, that $\Box p$ and $p$ may be irrelevant to one another reinforces an interpretation of the intensional $\Box$ as modifying the subject-matter of sentences to which it applies.

\subsection{Extensions of $\logic{B}^\Box$}

Now, we proceed to consider some extensions to $\mathsf{B}^{\Box}$ and consider hyperformality in such extensions. For $S\subseteq \{\mathrm{T},\mathrm{4},\mathrm{K},\mathrm{RN}\}$, let $\logic{B}^\Box_S$ be the logic we get by adding the members of $S$ as additional axioms/rules. To define hyperformalism for these logics, we define relations on $\positions{MPos}$ corresponding to T, 4, K, and RN by saying for all $\mathtt{\overline{x}}$ and $\mathtt{\overline{y}}$ in $\positions{MPos}$ that $\mathtt{\overline{x}b\overline{y}}\sim_{\mathrm{T}}\mathtt{\overline{xy}}$, that $\mathtt{\overline{x}b\overline{y}}\sim_{\mathrm{4}}\mathtt{\overline{x}bb\overline{y}}$, that $\mathtt{\overline{x}bl\overline{y}}\sim_{\mathrm{K}}\mathtt{\overline{x}lb\overline{y}}$ and $\mathtt{\overline{x}br\overline{y}}\sim_{\mathrm{K}}\mathtt{\overline{x}rb\overline{y}}$, and that $\mathtt{\overline{x}lb}\sim_{\mathrm{RN}}\mathtt{\overline{x}rb}$.

Where $S\subseteq \{\mathrm{T},\mathrm{4},\mathrm{K},\mathrm{RN}\}$, let $\sim_S$ be the smallest equivalence relation generated by $\sim_{\mathtt{nn}}\cup\bigcup_{s\in S}\sim_s$ and write $[\overline{\mathtt{x}}]_S$ for the $\sim_S$-equivalence class of $\overline{\mathtt{x}}$. Let $\positions{MPos}_S=\{[\overline{\mathtt{x}}]_S:\overline{\mathtt{x}}\in\positions{MPos}\}$ and $\positions{MPos}_S$-substitutions be functions $\sigma:\positions{MPos}_S\times\At\to\lang{L}$. $\positions{MPos}_S$-substitutions extend as indicated in the previous section. Say that a set of formulas $X$ is $\positions{MPos}_S$-hyperformal just if for all $A\in X$ and all $\positions{MPos}_S$-substitutions $\alpha$,    $\alpha_{[\varepsilon]_S}(A)\in X$. 





 We can now state the result for extensions of $\logic{B}^{\Box}$:

\begin{theorem}
    For each $S\subseteq \{\mathrm{T},\mathrm{4},\mathrm{K},\mathrm{RN}\}$, $\logic{B}^\Box_{S}$ is $\positions{MPos}_S$-hyperformal.
\end{theorem}
\begin{proof}
    By induction on derivations. Appealing to Theorem \ref{theorem:BBox-Hyperformal} leaves only the four cases to consider:
    
    For T, suppose $\mathrm{T}\in S$. Note that $\sigma_{[\varepsilon]_S}(\Box A\to A)=\Box\sigma_{[\mathtt{bc}]_S}(A)\to\sigma_{[\mathtt{c}]_S}(A)$. Now observe that since $\mathtt{bc}\sim_{\mathrm{T}}\mathtt{c}$, $\sigma_{[\mathtt{bc}]_S}(A)=\sigma_{[\mathtt{c}]_S}(A)$. So in fact $\Box\sigma_{[\mathtt{bc}]_S}(A)\to\sigma_{[\mathtt{c}]_S}(A)$ is another instance of T. 

    For 4, suppose that $\mathrm{4}\in S$. We show that $\sigma_{[\varepsilon]_{S}}(\Box A\to\Box\Box A)$ is a theorem. Some manipulation reveals this to be equivalent to $\Box\sigma_{[\mathtt{bc}]_{S}}(A)\to\Box\Box\sigma_{[\mathtt{bbc}]_{S}}(A)$. Since $\mathtt{bc}\sim_{4}\mathtt{bbc}$, $\sigma_{[\mathtt{bc}]_{S}}(A)=\sigma_{[\mathtt{bbc}]_{S}}(A)$, showing that this is itself an instance of $4$.
    
    For K, suppose $\mathrm{K}\in S$. We compute that $\sigma_{\mathtt{[\varepsilon]}_S}(\Box(A\to B)\to(\Box A\to \Box B))$ is
    \begin{displaymath}
        \Box(\sigma_{\mathtt{[lbc]}_S}(A)\to\sigma_{\mathtt{[rbc]}_S}(B))\to(\Box\sigma_{\mathtt{[blc]}_S}(A)\to\Box\sigma_{\mathtt{[brc]}_S}(B)).
    \end{displaymath} 
    Now observe that since $\mathtt{lbc}\sim_{\mathrm{K}}\mathtt{blc}$, $\sigma_{[\mathtt{lbc}]_S}(A)=\sigma_{[\mathtt{blc}]_S}(A)$, and since $\mathtt{rbc}\sim_{\mathrm{K}}\mathtt{brc}$, $\sigma_{[\mathtt{rbc}]_S}(A)=\sigma_{[\mathtt{brc}]_S}(A)$. So the offset formula is again an instance of K.

    For RN, assume that $\mathrm{RN}\in S$ and that $\Box A$ was derived from an application of RN to an earlier theorem $A$. Fix a substitution $\sigma$ and define a further substitution $\sigma^{\star\star}$ so that

    \begin{center}
        $\sigma^{\star\star}_\mathtt{[\overline{x}]_S}(p)=
        \begin{cases}
            \sigma_{[\mathtt{\overline{y}lb}]_S}(p) & \mbox{ if }\mathtt{\overline{x}}=\mathtt{\overline{y}c}\\
            \sigma_{]\mathtt{\overline{x}b}]_S}(p) & \mbox{ otherwise}
        \end{cases}$
    \end{center}

    \noindent Let us review three facts about $\sigma^{\star\star}$: First, one can check that $\sigma^{\star\star}$ is faithful. Second, as before, an easy induction establishes that the two cases of the definition remain correct for any formula $A$. Third, $\sigma^{\star\star}$ respects each relation $\sim_{S}$, \emph{i.e.}, if $\mathtt{\overline{x}}\sim_{S}\mathtt{\overline{y}}$ then $\sigma^{\star\star}_{[\mathtt{\overline{x}}]}(A)=\sigma^{\star\star}_{[\mathtt{\overline{y}}]}(A)$. Consequently, by induction hypothesis, $\sigma^{\star\star}_{[\varepsilon]_{S}}(A)$ is a theorem. Now, we consider two cases based on the form of $A$. If $A$ is of the form $B\rightarrow C$, then $\sigma^{\star\star}_{[\varepsilon]_{S}}(B\rightarrow C)=\sigma^{\star\star}_{[\mathtt{c}]_{S}}(B)\rightarrow\sigma^{\star\star}_{[\mathtt{c}]_{S}}(C)$. This by selection of $\sigma^{\star\star}$ is $\sigma_{[\mathtt{lb}]_{S}}(B)\rightarrow\sigma_{[\mathtt{lb}]_{S}}(C)$ and as $\mathtt{lb}\sim_{RN}\mathtt{rb}$, $\sigma_{[\mathtt{lb}]_{S}}(B)\rightarrow\sigma_{[\mathtt{rb}]_{S}}(C)=\sigma_{[\mathtt{b}]_{S}}(B\rightarrow C)$ is provable as well. By an application of RN, then, $\Box(\sigma_{[\mathtt{b}]_{S}}(B\rightarrow C))$ is a theorem, \emph{i.e.}, $\sigma_{[\varepsilon]_{S}}(\Box(B\rightarrow C))$ is a theorem as well. In case $A$ is not of the form $B\to C$, then by construction, $\sigma^{\star\star}_{[\varepsilon]_{S}}(A)=\sigma_{[\mathtt{b}]_{S}}(A)$. By an application of RN, $\Box(\sigma_{[\mathtt{b}]_{S}}(A))$ will be a theorem, whence $\sigma_{[\varepsilon]_{S}}(\Box A)$ will be a theorem as well.

    
\end{proof} 

\noindent As before, we can apply features of hyperformalism to demonstrate refined variable sharing properties enjoyed by relevant modal logics:

\begin{definition}
A logic in the language $\mathbf{L}$ enjoys the $\positions{MPos}_S$-variable sharing property in case whenever $A\rightarrow B$ is provable, there are $A[p]$ in $A$ and $B[p]$ in $B$ such that $\mathsf{pos}(A[p]\rightarrow B)\sim_{\mathtt{nn}}\mathsf{pos}(A\rightarrow B[p])$ and $\mathsf{pos}(A[p]\rightarrow B)\sim_{S}\mathsf{pos}(A\rightarrow B[p])$.
\end{definition}

\noindent We can provide a schematic characterization of the variable sharing properties holding of these systems:

\begin{theorem}
   Each logic $\logic{B}_{S}^{\Box}$ has the $\positions{MPos}_{S}$-variable sharing property.
\end{theorem}

\begin{proof}
    By analogous means to Theorem \ref{theorem:BBox-MPosVSP}, modifying the definition of $\mathsf{g}$ to require respect for $\sim_{S}$.
\end{proof}

\noindent Many such systems will have rather elegant and novel variable sharing properties, an example of which can be illustrated by the system $\logic{B}_{\lbrace 4\rbrace}^{\Box}$:

\begin{corollary}
    Say that an occurrence of an atom in a formula has \emph{extensional mode} if it is not within the scope of a $\Box$ and \emph{intensional mode} otherwise. Then if $A\rightarrow B$ is a theorem of $\logic{B}_{\lbrace 4\rbrace}^{\Box}$, then there is a variable occurring in both $A$ and $B$ with identical mode.
\end{corollary}

\noindent The inclusion of the axiom 4 eliminates the \emph{count} of instances of $\Box$ from  playing a role in determining a variable sharing property. (If translated into terms of aboutness, this means that including axiom 4 ensures what \cite{ferguson2023c} calls ``K-idempotence,'' \emph{i.e.}, that while the subject-matters of $A$ and $\Box A$ may differ, those of $\Box A$ and $\Box \Box A$ must coincide.) Axioms K or T, themselves, have consequences for the strength of the variable sharing properties holding of systems in which they are included.

\section{$\logic{K}^{\positions{MPos}}$}\label{section:KMPos}

$\logic{B}^\Box$ is a $\positions{MPos}$-hyperformal sublogic of the classical modal logic $\logic{K}$. We now turn to describing another $\positions{MPos}$-hyperformal sublogic of $\logic{K}$. We will call this logic $\logic{K}^{\positions{MPos}}$ and give three different descriptions of it. The first description is in terms of $\positions{MPos}$-varying Kripke models. The second is via a tableau system that we will show is both sound and complete for the varying Kripke models. The third description is a characterization of where $\logic{K}^{\positions{MPos}}$ sits among the $\positions{MPos}$-hyperformal subsets of $\logic{K}$. As we will see, it is the largest such set, which gives us good reason to study it further. Unfortunately, because of space constraints, that study will have to wait for another day.

\subsection{The Semantic Characterization}

Given a Kripke frame $K=\langle W,\rho\rangle$, a $\positions{MPos}^K$-assignment is a function $\alpha:\positions{MPos}\times W\times\At\to\{0, 1\}$. $\positions{MPos}^K$-assignments extend via the following recursive clauses:

\begin{itemize}
    \item $\alpha_{\overline{\mathtt{x}}}(\omega, A\land B)=\min(\alpha_{\overline{\mathtt{x}}}(\omega, A), \alpha_{\overline{\mathtt{x}}}(\omega, B))$;
    \item $\alpha_{\overline{\mathtt{x}}}(\omega, A\lor B)=\max(\alpha_{\overline{\mathtt{x}}}(\omega, A), \alpha_{\overline{\mathtt{x}}}(\omega, B))$;
    \item $\alpha_{\varepsilon}(\omega, A\to B)=\max(1-\alpha_{\mathtt{c}}(\omega, A), \alpha_{\mathtt{c}}(\omega, B))$;
    \item $\alpha_{\overline{\mathtt{x}}}(\omega, A\to B)=\max(1-\alpha_{\mathtt{l\overline{x}}}(\omega, A), \alpha_{\mathtt{r\overline{x}}}(\omega, B))$ for $\overline{x}\neq\varepsilon$;
    \item $\alpha_{\overline{\mathtt{x}}}(\omega, \neg A)=1-\alpha_{\mathtt{n}\overline{\mathtt{x}}}(\omega, A)$; and
    \item $\alpha_{\overline{\mathtt{x}}}(\omega, \Box A)=\inf\{\alpha_{\mathtt{b}\overline{\mathtt{x}}}(v, A):\langle \omega, v\rangle\in \rho\}$.
\end{itemize}

We say that $A$ is $\positions{MPos}$-valid (and we write $\vDash_{\positions{MPos}} A$) just if for all Kripke frames $K$, all $\positions{MPos}^K$-assignments $\alpha$, and all $\omega\in K$, $\alpha_{\varepsilon}(\omega, A)=1$. 

As a referee has helpfully pointed out, it is possible to see functions $\alpha:\positions{MPos}\times W\times\At\to\{0, 1\}$ as (somewhat) ordinary Kripke models with a funny set of worlds---in particular, with the set of worlds $\positions{MPos}\times W$ and with assignments that, as it were, `happen to' vary along one of the world-dimensions in regular ways. This sort of thing has been explored---albeit in a different setting---in \cite{logan2025content}, where some of the details were worked out. We think that this is a selling point for the logics at hand---a point that we will revisit in the conclusion. 

\subsection{The Tableaux Characterization}

\begin{figure}[tb]
    \begin{center}
        \begin{tabular}{|c|c|c|c|}\hline
            \xymatrix@C=-11mm@R=4mm{ 
                & \overline{\mathtt{x}},\omega,B_1\lor B_2,1\ar[dr]\ar[dl] & \\
                \overline{\mathtt{x}},\omega,B_1,1 & & \overline{\mathtt{x}},\omega,B_2,1
            }
            &
            \xymatrix@C=-10mm@R=4mm{
                \overline{\mathtt{x}},\omega,B_1\lor B_2,0\ar[d] \\
                \txt{$\overline{\mathtt{x}},\omega,B_1,0$ \\ $\overline{\mathtt{x}},\omega,B_2,0$} 
            }
            &
            \xymatrix@C=-10mm@R=4mm{
                \overline{\mathtt{x}},\omega,B_1\land B_2,1\ar[d] \\
                \txt{$\overline{\mathtt{x}},\omega,B_1,1$ \\ $\overline{\mathtt{x}},\omega,B_2,1$} 
            }
            &
            \xymatrix@C=-11mm@R=4mm{ 
                & \overline{\mathtt{x}},\omega,B_1\land B_2,0\ar[dr]\ar[dl] & \\
                \overline{\mathtt{x}},\omega,B_1,0 & & \overline{\mathtt{x}},\omega,B_2,0
            }
            \\ \hline
            \xymatrix@C=-11mm@R=4mm{
                & \varepsilon,\omega,B_1\to B_2\ar[dr]\ar[dl] & \\
                 \mathtt{c},\omega,B_1,0 & & \mathtt{c},\omega,B_2,1
            }
            &
            \xymatrix@C=-11mm@R=4mm{
                & \mathtt{\overline{y}},\omega,B_1\to B_2\ar[dr]\ar[dl] & \\
                 \mathtt{l\overline{y}},\omega,B_1,0 & & \mathtt{r\overline{y}},\omega,B_2,1
            }
            &
            \xymatrix@C=-10mm@R=4mm{
                \varepsilon,\omega,B_1\to B_2,0\ar[d] \\
                \txt{$ \mathtt{c},\omega,B_1,1$ \\ $ \mathtt{c},\omega,B_2,0$}
            }
            &
            \xymatrix@C=-10mm@R=4mm{
                \mathtt{\overline{y}},\omega,B_1\to B_2,0\ar[d] \\
                \txt{$ \mathtt{l\overline{y}},\omega,B_1,1$ \\ $\mathtt{r\overline{y}},\omega,B_2,0$}
            }
            \\ \hline
            \xymatrix@C=-10mm@R=4mm{
                \overline{\mathtt{x}},\omega,\neg B,1\ar[d] \\
                 \mathtt{n\overline{x}},\omega,B,0
            }
            &
            \xymatrix@C=-10mm@R=4mm{
                \overline{\mathtt{x}},\omega,\neg B,0\ar[d] \\
                 \mathtt{n\overline{x}},\omega,B,1
            }
            &

            \xymatrix@C=-10mm@R=4mm{
                \txt{$\omega_1 \R \omega_2$ \\ $\overline{\mathtt{x}}, \omega_1, \Box B, 1$}\ar[d]\\
		          \mathtt{b\overline{\mathtt{x}}}, \omega_2, B, 1
        	}
            &
            \xymatrix@C=-10mm@R=3mm{
		          \overline{\mathtt{x}}, \omega, \Box B, 0\ar[d] \\
                \txt{$\omega \R \nu$ \\ $\mathtt{b\overline{\mathtt{x}}}, \nu, B, 0$}
            }
            \\ \hline 
        \end{tabular}
    \end{center}
    \caption{Tableau rules; we assume $\overline{y}\neq\varepsilon$ and $\nu$ is fresh.}\label{tableaurules}
\end{figure}

$\positions{MPos}$-validity can be tracked by analytic tableaux that have nodes that are either quadruples $\langle\overline{\mathtt{x}}, \omega, A, i\rangle$ or $\R$-expressions of the form $\omega_i \R \omega_j$. In quadruples, the $\overline{x}$-index is drawn from $\positions{MPos}$ and tracks positions, the $\omega$-index is drawn from some set $\Omega=\{\omega_i\}_{i=1}^\infty$ used to track \textit{worlds}, $A$ is a formula and $i\in\{0,1\}$. $\R$-expressions, on the other hand, are used to track accessibility facts. 

Tableaux are governed by the rules in Figure~\ref{tableaurules}. A \emph{branch} on a tableau closes when it contains nodes of the form $\langle\mathtt{\overline{x}},\omega,A,1\rangle$ and $\langle\mathtt{\overline{y}},\omega,A,0\rangle$ for $\mathtt{\overline{x}}\sim\mathsf{\overline{y}}$. A \emph{tableau} closes when all of its branches close. We will write $\vdash_{\positions{MPos}} A$ to mean that there is a closed $S$-tableau that begins at $\langle\varepsilon, \omega_1, A, 0\rangle$.\footnote{One could, instead, add rules corresponding to the conditions the various logics impose on the relation $\R$ as is done in e.g. \cite{beall2017logic}. We've opted for the current approach to avoid painful, unilluminating discussions of what to do in our completeness proof when the rules governing $\R$ can be applied without bound.}

\subsubsection{Metatheory 1: Setup}

We will say that a \emph{branchset} is a set all of whose members are possible nodes on a tableau. So if $\beta$ is a branchset, then each member of $\beta$ either has the form $\langle\overline{\mathtt{x}},\omega,A,i\rangle$ with $\overline{\mathtt{x}}\in\positions{MPos}$, $\omega\in\Omega$, $A$ a formula, and $i\in\{0,1\}$ or has the form $\omega_i\R \omega_j$ with $\{\omega_i,\omega_j\}\subseteq\Omega$. Given a branchset $\beta$, let $\Omega_\beta$ be the subset of $\Omega$ containing all $\omega_i$ that occur in some member of $\beta$. A \emph{correspondence} is a triple consisting of a Kripke frame $K=\langle W,\rho\rangle$, a function $f:\Omega_\beta\to W$, and an $\positions{MPos}^{K}$ assignment $\alpha$. Given a correspondence $\kappa=\langle K,f,\alpha\rangle$ with $K=\langle W_K,\rho_K\rangle$, $\kappa$ is said to \emph{conform to} $\beta$ when $\alpha_{\overline{\mathtt{x}}}(f(\omega),A)=i$ if $\langle\overline{\mathtt{x}},\omega,A,i\rangle\in\beta$, and $\langle f(\omega_i),f(\omega_j)\rangle\in\rho_K$ if $\omega_i\R \omega_j\in\beta$. 

Given a branchset $\beta$, we define the $\langle\overline{\mathtt{x}},\omega,A,i\rangle$-extensions of $\beta$ by cases as follows. First, $\beta$ is the only $\langle\overline{\mathtt{x}},\omega,A,i\rangle$-extension of $\beta$ if $\langle\overline{\mathtt{x}},\omega,A,i\rangle\not\in\beta$. Otherwise, 
\begin{itemize}
    \item if $A$ is an atom, then $\beta$ is the only $\langle\overline{\mathtt{x}},\omega,A,i\rangle$-extension of $\beta$.
    \item Both $\beta\cup\{\langle\overline{\mathtt{x}},\omega,B_1,1\rangle\}$ and $\beta\cup\{\langle\overline{\mathtt{x}},\omega,B_2,1\rangle\}$ are $\langle\overline{\mathtt{x}},\omega,B_1\lor B_2,1\rangle$-extensions of $\beta$.
    \item $\beta\cup\{\langle\overline{\mathtt{x}},\omega,B_1,0\rangle,\langle\overline{\mathtt{x}},\omega,B_2,0\rangle\}$ is the only $\langle\overline{\mathtt{x}},\omega,B_1\lor B_2,0\rangle$-extension of $\beta$.
    \item $\beta\cup\{\langle\overline{\mathtt{x}},\omega,B_1,1\rangle,\langle\overline{\mathtt{x}},\omega,B_2,1\rangle\}$ is the only $\langle\overline{\mathtt{x}},\omega,B_1\land B_2,1\rangle$-extension of $\beta$.
    \item Both $\beta\cup\{\langle\overline{\mathtt{x}},\omega,B_1,0\rangle\}$ and $\beta\cup\{\langle\overline{\mathtt{x}},\omega,B_2,0\rangle\}$ are $\langle\overline{\mathtt{x}},\omega,B_1\land B_2,0\rangle$-extensions of $\beta$.
    \item Both $\beta\cup\{\langle\mathtt{c},\omega,B_1,0\rangle\}$ and $\beta\cup\{\langle\mathtt{c},\omega,B_2,1\rangle\}$ are $\langle\varepsilon,\omega,B_1\to B_2,1\rangle$-extensions of $\beta$.
    \item $\beta\cup\{\langle\mathtt{c},\omega,B_1,1\rangle,\langle\mathtt{c},\omega,B_2,0\rangle\}$ is the only $\langle\overline{\mathtt{x}},\omega,B_1\to B_2,0\rangle$-extension of $\beta$.
    \item For $\overline{\mathtt{x}}\neq\varepsilon$, both $\beta\cup\{\langle\mathtt{l\overline{x}},\omega,B_1,0\rangle\}$ and $\beta\cup\{\langle\mathtt{r\overline{x}},\omega,B_2,1\rangle\}$ are $\langle\varepsilon,\omega,B_1\to B_2,1\rangle$-extensions of $\beta$.
    \item For $\overline{\mathtt{x}}\neq\varepsilon$, $\beta\cup\{\langle\mathtt{l\overline{x}},\omega,B_1,1\rangle,\langle\mathtt{r\overline{x}},\omega,B_2,0\rangle\}$ is the only $\langle\overline{\mathtt{x}},\omega,B_1\to B_2,0\rangle$-extension of $\beta$.
    \item $\beta\cup\{\langle\mathtt{n\overline{x}},\omega,B,0\rangle\}$ is the only $\langle\overline{\mathtt{x}},\omega,\neg B,1\rangle$-extension of $\beta$.
    \item $\beta\cup\{\langle\mathtt{n\overline{x}},\omega,B,1\rangle\}$ is the only $\langle\overline{\mathtt{x}},\omega,\neg B,0\rangle$-extension of $\beta$.
    \item $\beta\cup\{\langle\mathtt{b\overline{x}},\omega',B,1\rangle: \omega\R \omega'\in\beta\}$ is the only $\langle\overline{\mathtt{x}},\omega,\Box B,1\rangle$-extension of $\beta$.
    \item For each $\omega'\in\Omega-\Omega_\beta$,  $\beta\cup\{\langle\mathtt{b\overline{x}},\omega',B,0\rangle, \omega\R\omega'\}$ is an $\langle\overline{\mathtt{x}},\omega,\Box B,0\rangle$-extension of $\beta$
\end{itemize}\smallskip

\subsubsection{Metatheory 2: Soundness}

\begin{lemma}
    Let $\beta$ be a branchset, $\kappa=\langle K,f,\alpha\rangle$ be a correspondence, $K=\langle W,\rho\rangle$ and $\langle\mathtt{\overline{x}},\omega,A,i\rangle\in\beta$. Then if $\kappa$ conforms to $\beta$ then for some $\langle\mathtt{\overline{x}},\omega,A,i\rangle$-extension $\beta'$ of $\beta$, there is $g:\Omega_{\beta'}\to W$ such that $g(\omega)=f(\omega)$ for all $\omega\in\Omega_\beta$ and $\kappa_g=\langle K,g,\alpha\rangle$ conforms to $\beta'$.
\end{lemma}
\begin{proof}
    By examining cases. If $A$ is an atom, the result is trivial. For the non-modal connectives, the arguments go as in \cite{standefer2024topics}. For the modal cases, we argue as follows.

    Suppose $A=\Box B$ and $i=1$. If for no $\omega'$ is $\omega\R\omega'\in\beta$, then the result is trivial. So suppose $\omega\R\omega'\in\beta$. Then since $\kappa$ conforms to $\beta$, $\alpha_{\mathtt{\overline{x}}}(f(\omega),\Box B)=\inf\{\alpha_{\mathtt{b}\overline{\mathtt{x}}}(v, B):\langle f(\omega), v\rangle\in \rho\}=1$ and  $\langle f(\omega),f(\omega')\rangle\in\rho$. So $\alpha_{\mathtt{b\overline{x}}}(f(\omega'),B)=1$, and thus $\kappa$ itself conforms to $\beta\cup\{\langle\mathtt{b\overline{x}},\omega',B,1\rangle: \omega\R \omega'\in\beta\}$.

    Suppose $A=\Box B$ and $i=0$. Since $\kappa$ conforms to $\beta$, $\alpha_{\mathtt{\overline{x}}}(f(\omega),\Box B)=\inf\{\alpha_{\mathtt{b}\overline{\mathtt{x}}}(v, B):\langle f(\omega), v\rangle\in \rho\}=0$. So for some $v\in W$, $\alpha_{\mathtt{b}\overline{\mathtt{x}}}(v, B)=0$ and $\langle f(\omega),v\rangle\in\rho$. Choose $\omega'\in\Omega-\Omega_\beta$ and let $g(\omega)=f(\omega)$ for all $\omega\in\Omega_\beta$ and $g(\omega')=v$. Then clearly $\kappa_g$ conforms to $\beta\cup\{\langle\mathtt{b\overline{x}},\omega',B,0\rangle, \omega\R\omega'\}$.
\end{proof}

\begin{theorem}
    If $\vdash_{\positions{MPos}} A$, then $\vDash_{\positions{MPos}} A$.
\end{theorem}
\begin{proof}
    Let $\vdash_{\positions{MPos}} A$. Suppose $K=\langle W,\rho\rangle$ is a Kripke frame and $\alpha$ is an $\positions{MPos}^K$-assignment. If $\alpha_\varepsilon(w,A)=0$, then letting $f(\omega_1)=w$, the triple $\langle K,f,\alpha\rangle$ conforms to $\langle\varepsilon\omega_1,A,0\rangle$. So by the previous lemma, it conforms to all of its extensions. But since $\vdash_{\positions{MPos}} A$, all such extensions eventually contain nodes of the form $\langle\overline{x},\omega_j,A,1\rangle$ and $\langle\overline{x},\omega_j,A,0\rangle$ and thus cannot be conformed to. So for no $K$ and $\alpha$ do we have $\alpha_\varepsilon(w,A)=0$. So $\vDash_{\positions{MPos}_S} A$. 
\end{proof}

\subsubsection{Metatheory 3: Completeness}

We have, to this point, been implicitly restricting our attention to finite branchsets. But (medicine and forestry aside) there's no good reason to stick to this restriction.\footnote{From here, what we say follows the contours of the argument in \cite{FittingMendelsohn}.} Given a finite branchset $\beta$, say that a $\beta$-tableau is a tableau that begins with the members of $\beta$. Given \emph{any} (finite or infinite) branchset $\beta$, say that $\beta$ is \emph{tableau-inconsistent} when for some finite $\beta'\subseteq\beta$, there is a closed $\beta'$-tableau. A branchset is \emph{tableau-consistent} if it's not tableau inconsistent. A branchset is \emph{maximally} tableau-consistent when it is tableau-consistent, but all proper tableau extensions of it are tableau inconsistent. A branchset is \emph{witnessed} when it contains $\langle\mathtt{\overline{x}},\omega,\Box B,0\rangle$ only if for some $\omega'$ it contains both $\omega\R\omega'$ and $\langle\mathtt{b\overline{x}},\omega', B,0\rangle$.

\begin{lemma}
    If $\beta$ is tableau-consistent and $\Omega_\beta$ is finite, then there is a maximal, witnessed, tableau-consistent extension $\gamma$ of $\beta$.
\end{lemma}
\begin{proof}

    Enumerate the set of possible branchset entries as $Z_1,Z_2,\dots$ Note that some of these are quadruples and some are $\R$-claims. Define a sequence of sets $\gamma_1,\gamma_2,\dots$ as follows.

    \begin{itemize}
        \item $\gamma_1=\beta$
        \item If $\gamma_i\cup\{Z_i\}$ is tableau-inconsistent, then $\gamma_{i+1}=\gamma_i$.
        \item Otherwise, if $Z_i$ is a quadruple $\langle\mathtt{\overline{x}},\omega_j,\Box B,0\rangle$, then $\gamma_{i+1}=\gamma_i\cup\{Z_i,\langle\mathtt{\overline{x}},\omega_k, B,0\rangle,\omega_j\R\omega_k\}$ where $k$ is the least $n$ such that $\omega_n$ does not occur in $\gamma_i$.
        \item Otherwise $\gamma_{i+1}=\gamma_i\cup\{Z_i\}$.
    \end{itemize}
    Let $\gamma=\cup_{i=1}^\infty\gamma_i$. I claim that $\gamma$ is a maximal, witnessed, tableau-consistent extension of $\beta$. Since it's obvious that $\gamma$ is an extension of $\beta$, there are three things for us to check:
    \begin{description}[leftmargin=*]
        \item[$\gamma$ is maximal:] Pick a possible branchset entry $E$. Then $E=Z_i$ for some $i$. By construction, either $E\in\gamma_{n+1}\subseteq\gamma$ or $\gamma_n\cup\{E\}$ is tableau-inconsistent. But in the latter case, $\gamma\cup\{E\}$ is obviously tableau-inconsistent as well. So if $E\not\in\gamma$, then $\gamma_n\cup\{E\}$ is tableau-inconsistent, and thus $\gamma$ is maximal.
        \item[$\gamma$ is tableau-consistent:] We prove that each $\gamma_n$ is tableau-consistent. It follows from this that $\gamma$ is tableau-consistent, but we will also want to know this below.

        The proof is by contradiction. Suppose for some $n$ that $\gamma_n$ is tableau-inconsistent and let $\gamma_k$ be the least such $n$. Since $\beta$ is tableau-consistent, $k\neq 1$, so $\gamma_{k-1}$ and $Z_{k-1}$ are defined. We consider three cases:
        \begin{itemize}
            \item If $\gamma_{k-1}\cup\{Z_{k-1}\}$ is tableau-inconsistent, then $\gamma_{k-1}=\gamma_k$. But then since $\gamma_k$ is tableau-inconsistent, so is $\gamma_{k-1}$, contradicting $k$'s minimality. So $\gamma_{k-1}\cup\{Z_{k-1}\}$ is tableau-consistent.
            \item If $Z_{k-1}$ is $\langle\mathtt{\overline{x}},\omega_j,\Box B,0\rangle$, then $\gamma_{k}=\gamma_{k-1}\cup\{Z_{k-1},\langle\mathtt{\overline{x}},\omega_l, B,0\rangle,\omega_j\R\omega_k\}$ where $l$ is the least $m$ such that $\omega_m$ does not occur in $\gamma_{k-1}$. 

            Since $\gamma_k$ is tableau-inconsistent, there is a finite subset $F$ of $\gamma_{k-1}$ and a closed tableau that starts at $F\cup\{Z_{k-1},\langle\mathtt{\overline{x}},\omega_l, B,0\rangle,\omega_j\R\omega_k\}$. But then since $\omega_j$ is fresh to $F\cup\{Z_{k-1}\}$, the very same tableau is a closed tableau that starts at $F\cup\{Z_{k-1}\}$. But this contradicts the fact that $\gamma_{k-1}\cup\{Z_{k-1}\}$ is tableau-consistent. 
            \item Otherwise $\gamma_{k}=\gamma_{k-1}\cup\{Z_{k-1}\}$ and $\gamma_{k-1}\cup\{Z_{k-1}\}$ is consistent, which is a contradiction.
        \end{itemize}
        Thus, for no $n$ is $\gamma_n$ tableau-inconsistent. So $\gamma$ is tableau-consistent.
        \item[$\gamma$ is witnessed:] Suppose $\gamma$ contains $\langle\mathtt{\overline{x}},\omega,\Box B,0\rangle$. For some $n$, $\langle\mathtt{\overline{x}},\omega,\Box B,0\rangle=Z_n$. Since $\gamma_n$ is consistent, there is an $\omega'$ so that $\gamma_{n+1}$---and thus $\gamma$---contains contains both $\omega\R\omega'$ and $\langle\mathtt{b\overline{x}},\omega', B,0\rangle$.
    \end{description}
\end{proof}

Given a maximal, witnessed, tableau-consistent branchset $\beta$, the canonical $\beta$-correspondence $\kappa_\beta$ is the correspondence $\langle K,\iota,\alpha\rangle$ where $K=\langle\Omega_\beta,\{\langle\omega_i,\omega_j\rangle:\omega_1\R\omega_j\in\beta\}\rangle$, $\iota$ is the identity function, and $\alpha_{\mathtt{\overline{x}}}(\omega,p)=1$ iff $\langle\mathtt{\overline{x}},\omega,p,1\rangle\in\beta$.

\begin{lemma}
    Let $\beta$ be a maximal, witnessed, tableau-consistent branchset and let $\kappa_\beta=\langle K,\iota,\alpha\rangle$. If $\langle\mathtt{\overline{x}},\omega,A,1\rangle\in\beta$ then $\alpha_{\mathtt{\overline{x}}}(\omega,A)=1$ and if $\langle\mathtt{\overline{x}},\omega,A,0\rangle\in\beta$ then $\alpha_{\mathtt{\overline{x}}}(\omega,A)=0$
\end{lemma}
\begin{proof}
    By induction on $A$. If $A$ is an atomic $p$, then by the definition of the canonical $\beta$-correspondence, $\langle\mathtt{\overline{x}},\omega,p,1\rangle\in\beta$ iff $\alpha_{\mathtt{\overline{x}}}(\omega,p)=1$ and since $\beta$ is tableau-consistent, it follows that if $\langle\mathtt{\overline{x}},\omega,A,0\rangle\in\beta$ then $\alpha_{\mathtt{\overline{x}}}(\omega,A)=0$. Of the remaining cases, we consider the following:

    Suppose $\langle\mathtt{\overline{x}},\omega,A\land B,1\rangle\in\beta$. If $\beta\cup\{\langle\mathtt{\overline{x}},\omega,A,1\rangle\}$ is tableau-inconsistent, then there is finite $F\subseteq\beta$ so that $F\cup\{\langle\mathtt{\overline{x}},\omega,A,1\rangle\}$ closes. But then since $\langle\mathtt{\overline{x}},\omega,A\land B,1\rangle\in\beta$, $F\cup\{\langle\mathtt{\overline{x}},\omega,A\land B,1\rangle\}$ is also a finite subset of $\beta$ that must also close, contradicting the fact that $\beta$ is tableau-consistent. So $\beta\cup\{\langle\mathtt{\overline{x}},\omega,A\land B,1\rangle\}$ is tableau-consistent. Thus by maximality, we must have that $\langle\mathtt{\overline{x}},\omega,A\land B,1\rangle\in\beta$. Mutatis mutandis, the same argument works for $B$. By the inductive hypothesis, since $\{\langle\mathtt{\overline{x}},\omega,A,1\rangle,\langle\mathtt{\overline{x}},\omega,B,1\rangle\}\subseteq\beta$, $\alpha_{\mathtt{\overline{x}}}(\omega,A)=\alpha_{\mathtt{\overline{x}}}(\omega,B)=1$, and thus $\alpha_{\mathtt{\overline{x}}}(\omega,A\land B)=1$ as well.

    Suppose $\langle\varepsilon,\omega,A\to B,1\rangle\in\beta$. By a similar argument to the one in the $\land$-case, we can see that if both $\beta\cup\{\langle\mathtt{c},\omega,A,0\rangle\}$ and $\beta\cup\{\langle\mathtt{c},\omega,B,1\rangle\}$ are tableau-inconsistent, then $\beta$ must be as well. So one of these sets must be tableau-consistent. Thus by the fact that $\beta$ is maximal, we can see that either $\langle\mathtt{c},\omega,A,0\rangle\in\beta$ or $\langle\mathtt{c},\omega,B,1\rangle\in\beta$. Either way, by the inductive hypothesis we get that $\alpha_{\varepsilon}(\omega,A\to B)=1$

    Suppose $\langle\mathtt{\overline{x}},\omega,\Box A,0\rangle\in\beta$. Since $\beta$ is witnessed, it also contains both $\omega\R\omega'$ and $\langle\mathtt{b\overline{x}},\omega',B,0\rangle$ for some $\omega'$. By the inductive hypothesis, $\alpha_{\mathtt{b\overline{x}}}(\omega',B)=0$. And since $\omega\R\omega'\in\beta$, it follows that $\alpha_{\mathtt{\overline{x}}}(\omega,\Box B)=0$.
\end{proof}

\begin{theorem}
    If $\vDash_{\positions{MPos}} A$, then $\vdash_{\positions{MPos}} A$.
\end{theorem}
\begin{proof}
    Suppose $\not\vdash_{\positions{MPos}} A$. Let $\beta=\{\langle\varepsilon,\omega_0,A,0\rangle\}$. By Lemma 3.3, there is a maximal, witnessed, tableau-consistent extension $\beta'$ of $\beta$. So by Lemma 3.4, the canonical correspondence gives us a model that falsifies $A$ at $\omega$ and $\varepsilon$. So $\not\vDash_{\positions{MPos}} A$. 
\end{proof}

\subsection{The Hyperformal Subset Characterization}

For any set of formulas $X$, we define $\positions{MPos}(X)$ as follows:
\begin{displaymath}
    \positions{MPos}(X)=\{A:\sigma_\varepsilon(A)\in X\text{ for all }\positions{MPos}\text{-substitutions }\sigma\}
\end{displaymath}
One can prove, by following exactly the same arguments as in Lemma 5 of \cite{standefer2024topics} or Theorem 3 of \cite{leach2024logic} that $\positions{MPos}(X)$ is the largest subset of $X$ closed under $\positions{MPos}$ substitutions. 

\begin{theorem}
    The following are equivalent:
    \begin{itemize}
        \item $A\in\positions{MPos}(K)$.
        \item For some atomic and injective-up-to-faithfulness $\iota$, $\iota(A)\in\logic{K}$.
        \item $A\in\logic{K}^{\positions{MPos}}$.
    \end{itemize}
\end{theorem}

The proof of this result requires only the very mildest of changes from the proofs of similar results in \cite{standefer2024topics} or \cite{leach2024logic}, and is thus left to the reader. Note that this gives us a third characterization of $\logic{K}^{\positions{MPos}}$: it is the largest $\positions{MPos}$-hyperformal subset of $\logic{K}$.

We will end by proving that $\logic{K}^{\positions{MPos}}$ has the $\positions{MPos}$-variable sharing property. We need a bit of setup first. To begin, we identify the subsets $\positions{MPos}_+$ and $\positions{MPos}_-$ of $\positions{MPos}$ as follows:
    \begin{itemize}
        \item $\varepsilon\in\positions{MPos}_+$.
        \item If $\mathtt{\overline{x}}\in\positions{MPos}_+$, then $\mathtt{l\overline{x}}\in\positions{MPos}_-$, $\mathtt{r\overline{x}}\in\positions{MPos}_+$, $\mathtt{n\overline{x}}\in\positions{MPos}_-$, and $\mathtt{b\overline{x}}\in\positions{MPos}_+$.
        \item If $\mathtt{\overline{x}}\in\positions{MPos}_-$, then $\mathtt{l\overline{x}}\in\positions{MPos}_+$, $\mathtt{r\overline{x}}\in\positions{MPos}_-$, $\mathtt{n\overline{x}}\in\positions{MPos}_+$, and $\mathtt{b\overline{x}}\in\positions{MPos}_-$.
    \end{itemize}

Let $K$ be a one-world model with the universal accessibility relation. 

Suppose $A$ and $C$ share no atoms. Define the $\positions{MPos}^K$-assignments $f^+_A$ and $f^-_A$ as follows:
    \begin{align*}
    f^+_{A,\mathtt{\overline{y}}}(w,p) & =\left\{
        \begin{array}{rl} 
            1 & \text{if $\mathtt{\overline{y}=\overline{x}c}$, $\mathtt{\overline{x}}\in\positions{MPos}_+$, and $p$ occurs at $\mathtt{\mathsf{\overline{x}c}}$ in $A$} \\
            0 & \text{otherwise}
        \end{array}\right. \\
        f^-_{A,\mathtt{\overline{y}}}(w,p)& =\left\{
        \begin{array}{rl} 
            0 & \text{if $\mathtt{\overline{y}=\overline{x}c}$, $\mathtt{\overline{x}}\in\positions{MPos}_+$, and $p$ occurs at $\mathtt{\overline{x}c}$ in $A$} \\
            1 & \text{otherwise}
        \end{array}\right. \\
    \end{align*}

\begin{lemma}
    If $B$ is a subformula of $A$, then 
    \begin{itemize}
    	\item $f^+_{A,\mathtt{\overline{x}c}}(w,B)=1$ if $\mathtt{\overline{x}}\in\positions{MPos}_+$ and $B$ occurs in $A\to C$ at $\mathtt{\overline{x}c}$,
		\item $f^+_{A,\mathtt{\overline{x}c}}(w,B)=0$ if $\mathtt{\overline{x}}\in\positions{MPos}_-$ and $B$ occurs in $A\to C$ at $\mathtt{\overline{x}c}$,
		\item $f^-_{A,\mathtt{\overline{x}c}}(w,B)=0$ if $\mathtt{\overline{x}}\in\positions{MPos}_+$ and $B$ occurs in $C\to A$ at $\mathtt{\overline{x}c}$,
		\item $f^+_{A,\mathtt{\overline{x}c}}(w,B)=1$ if $\mathtt{\overline{x}}\in\positions{MPos}_-$ and $B$ occurs in $C\to A$ at $\mathtt{\overline{x}c}$.
	\end{itemize}
\end{lemma}
\begin{proof}
    By induction on $B$. Every case except the $\Box$-case goes exactly as in the proof of Lemmas 11 and 12 in \cite{standefer2024topics}. For that case, since $K$ has exactly one world and the universal accessibility relation, $f^+_{A,\mathtt{\overline{x}c}}(\Box B)=f^+_{A,\mathtt{\overline{x}c}}(B)$ and $f^-_{A,\mathtt{\overline{x}c}}(\Box B)=f^-_{A,\mathtt{\overline{x}c}}(B)$. So the inductive hypothesis immediately finishes the job.
\end{proof}

\begin{lemma}
    Suppose that $A$ and $B$ share no atoms. Define the assignment $h$ by
    \begin{displaymath}
        h_{\mathtt{\overline{x}}}(w,p)=\left\{
        \begin{array}{rl}
            f^+_{A,\mathtt{\overline{x}}}(w, p) & \text{ if } p \text{ occurs in $A$ and occurs at $\mathtt{\overline{x}}$ in $A\to B$} \\
            f^-_{B,\mathtt{\overline{x}}}(w, p) & \text{ if } p \text{ occurs in $A$ and occurs at $\mathtt{\overline{x}}$ in $A\to B$} \\
            \text{whatever} & \text{ otherwise}.
        \end{array}\right.
    \end{displaymath}
    Then if $C$ occurs at $\mathtt{\overline{x}}$ in $A$ then $h_{\mathtt{\overline{x}}}(w,C)=f^+_{A,\mathtt{\overline{x}}}(w,C)$ and if $C$ occurs at $\mathtt{\overline{x}}$ in $B$ then $h_{\mathtt{\overline{x}}}(w,C)=f^+_{B,\mathtt{\overline{x}}}(w,C)$.
\end{lemma}
\begin{proof}
    By induction on $C$, separately for each conclusion. 
\end{proof}

\begin{theorem}
    $\logic{K}^{\positions{MPos}}$ enjoys the variable sharing property. 
\end{theorem}
\begin{proof}
    Suppose $A$ and $B$ don't share atoms. Define $h$ as above. By the preceding lemmas, $h_{\mathtt{c}}(w,A)=1$ and $h_{\mathtt{c}}(w,B)=0$, so $h_{\varepsilon}(w, A\to B)=0$.
\end{proof}

\begin{corollary}
    $\logic{K}^{\positions{MPos}}$ has the $\positions{MPos}$-variable sharing property.
\end{corollary}
\begin{proof}
    Since $\logic{K}^{\positions{MPos}}$ is $\positions{MPos}$-hyperformal, this follows immediately from the previous theorem.
\end{proof}

    \section{Future Work and Concluding Remarks}

    We've characterized $\logic{K}^{\positions{MPos}}$ semantically and via a tableaux system. To better understand it, one would like an axiom system.\footnote{Thanks to a referee for pushing us on this point.} But to date logics formed by restricting known logics to their largest hyperformal (in some sense) fragments have resisted axiomatization. The only thing that \emph{is} known on this front is the rather unsatisfying result proved in \cite{leach2024logic}: for logics like these \emph{admit} recursive axiomatizations. Producing intelligible axiomatizations remains, in spite of this, as an outstanding project.


    $\logic{B}^\Box$ is $\positions{MPos}$-hyperformal and we can use this to extract a rather shockingly strong variable sharing result for its theorems. But it's not the strongest logic with these features, nor is it the largest subset of classical logic with these features. $\logic{K}^{\positions{MPos}}$ holds that honor. 

    Perhaps there are very good reasons for preferring $\logic{B}^\Box$ to $\logic{K}^{\positions{MPos}}$. If so, they ought to be enunciated, and we now know some things they \emph{can't} be: they can't be related to variable sharing or subclassicality. And one must admit that semantically, $\logic{K}^{\positions{MPos}}$ definitely has advantages over $\logic{B}^\Box$. Also note that since $\logic{B}^\Box$ is subclassical and closed under $\positions{MPos}$-substitutions, it is contained (properly as it turns out, though we leave it to the reader to verify this) in $\logic{K}^{\positions{MPos}}$. So anything you might want to do with a logic like $\logic{B}^\Box$, you can also do with $\logic{K}^{\positions{MPos}}$. 

    Another thing worth noting: ``relevant'' logics defined by pasting a simple filter on top of classical logic (a `syntactic sieve', as pejoratively described by Richard Sylvan in \cite{routley1982}), can be easily seen to incur nontransitivity in the sense that they will, for some $A$, $B$, and $C$, include $A\to B$ and $B\to C$ without including $A\to C$. Were this true for $\logic{K}^{\positions{MPos}}$, it would be a fairly serious reason to set it aside. But while $\logic{K}^{\positions{MPos}}$ \emph{is} the largest sublogic of $\logic{K}$ with a given feature, and it is \emph{a} sublogic of $\logic{K}$ with the $\positions{MPos}$-variable sharing property, $\logic{K}^{\positions{MPos}}$ is \emph{not} the largest sublogic of $\logic{K}$ with the $\positions{MPos}$-variable sharing property.
    
    As an example, consider the contraction axiom: $(p\to(p\to q))\to(p\to q)$. The leftmost $p$ and the $p$ antecedent of the terminal conditional both occur at the position $\positions{lc}$. So this formula \emph{does} have the $\positions{MPos}$-variable sharing property. But it isn't in $\logic{K}^{\positions{MPos}}$ because $\logic{K}^{\positions{MPos}}$ isn't defined by its possessing a given variable sharing property but by its possessing a given closure property---namely, closure under $\positions{MPos}$-substitutions. Since not all $\positions{MPos}$-substitutions map contraction to a theorem of classical logic, contraction isn't a member of $\logic{K}^{\positions{MPos}}$. 
    
    So $\logic{K}^{\positions{MPos}}$ isn't plausibly characterized as $\logic{K}$ with a syntactic sieve attached, and thus isn't subject to any of the usual criticisms that attach to such logics. And there's more: not only can't it be criticized on those fronts, it also doesn't \emph{suffer} on those fronts. In particular, we can show the following: 
    \begin{theorem}
        $\logic{K}^{\positions{MPos}}$ is transitive.
    \end{theorem}
    \begin{proof}
        Suppose $\sigma_\varepsilon(A\to C)=0$. Then $\sigma_{\positions{c}}(A)=1$ and $\sigma_{\positions{c}}(C)=0$. Now observe that if $\sigma_{\positions{c}}(B)=1$, then $\sigma_\varepsilon(B\to C)=0$ and if $\sigma_{\positions{c}}(B)=0$, then $\sigma_\varepsilon(A\to B)=0$. So if $A\to C\not\in\logic{K}^{\positions{MPos}}$, then either $A\to B\not\in\logic{K}^{\positions{MPos}}$ or $B\to C\not\in\logic{K}^{\positions{MPos}}$. Contraposing, if $A\to B\in\logic{K}^{\positions{MPos}}$ and $B\to C\in\logic{K}^{\positions{MPos}}$, then $A\to C\in\logic{K}^{\positions{MPos}}$.
    \end{proof}

\noindent This observation, then, cuts against a charge of artificiality in the formulation of $\logic{K}^{\positions{MPos}}$ insofar as the pathologies that follow from truly artificial approaches to relevance---like syntactic sieves---do not haunt the present system.

    All told, we would like it if this paper was understood as a challenge. Those who would prefer `traditional' relevant modal logics are being asked to give reasons for that preference. Absent such, we'd recommend that folks in need of relevant modal logics look to logics like $\logic{K}^{\positions{MPos}}$ rather than to logics like $\logic{B}^\Box$.



\end{document}